\documentclass[12pt]{article}
\usepackage{amsmath,amstext,amsfonts,amsthm}
\usepackage{amssymb,hyperref,verbatim}
\usepackage{marvosym,bbm}
\usepackage[T1]{fontenc}
\usepackage{stackengine}

\newtheorem{thm}{Theorem}
\newtheorem{lemma}{Lemma}[section]
\newtheorem{prop}[lemma]{Proposition}
\newtheorem{rmk}[lemma]{Remark}
\newtheorem{cor}[lemma]{Corollary}
\newtheorem{defin}[lemma]{Definition}

\def\tr{\operatorname{tr}}
\numberwithin{equation}{section}

\begin{document}

\title{On the Wegner orbital model}

\author{
Ron Peled\textsuperscript{1}, Jeffrey Schenker\textsuperscript{2},
Mira Shamis\textsuperscript{3},  Sasha Sodin\textsuperscript{4}}
\footnotetext[1]{School of Mathematical Sciences, Tel Aviv
University, Tel Aviv, 69978, Israel. E-mail:
peledron@post.tau.ac.il. Supported in part by ISF grants 1048/11 and 861/15,  IRG
grant SPTRF and the European Research Council start-up grant 678520
(LocalOrder).} \footnotetext[2]{Department of Mathematics, Michigan
State University, East Lansing, MI 48824. \mbox{E-mail}:
jeffrey@math.msu.edu. Supported in part by The Fund For Math and NSF
grant DMS-1500386.} \footnotetext[3]{Department of Mathematics, The
Weizmann Institute of Science, Rehovot 7610001, Israel. E-mail:
mira.shamis@weizmann.ac.il. Supported in part by ISF grant 147/15.} 
\footnotetext[4]{School of Mathematical
Sciences, Tel Aviv University, Tel Aviv, 69978, Israel. E-mail:
sashas{\MVOne}@post.tau.ac.il. Supported in part by the European
Research Council start-up grant 639305 (SPECTRUM).} \maketitle

\begin{abstract}
 The Wegner orbital model is a class of random operators introduced by Wegner
 to model the motion of a quantum particle with many internal degrees of freedom (orbitals)
 in a disordered medium. We consider the case when the matrix potential is Gaussian, and prove
 three results:
 localisation at strong disorder, a
Wegner-type estimate on the mean  density  of eigenvalues,
 and a Minami-type estimate on the probability of having multiple eigenvalues in a short interval.
 The last two results are proved in the more general
setting of deformed block-Gaussian matrices, which includes a class of
Gaussian band matrices as a special case. Emphasis is placed on the dependence
of the bounds on the number of orbitals.  As an additional application, we improve the
upper bound on the localisation length for one-dimensional Gaussian band matrices.
\end{abstract}

\section{Statement of results}

The current investigation is motivated by the work of
Wegner \cite{W1} and its continuation by Sch\"afer and Wegner
\cite{SW} and Oppermann and Wegner \cite{OW}
on the motion of a quantum particle with many
($N \gg 1$) internal degrees of freedom in
a disordered medium.

The Hamiltonian $H$ of the quantum particle acts on a dense subset
of $\ell_2(\mathbb{Z}^d \to \mathbb{C}^N)$, the space of
square-integrable functions from $\mathbb{Z}^d$ to $\mathbb{C}^N$,
via
\begin{equation}\label{eq:Z_d_model}
 (H\psi)(x) = V(x) \psi(x) + \sum_{y\, : \, y
\sim x} W(x, y) \psi(y)~, \quad x \in \mathbb{Z}^d~,
\end{equation}
where the potential entries $V(x)$ are $N \times N$ Hermitian
matrices, and the hopping terms $W(x, y)$ are $N \times N$ matrices
with the Hermitian constraint
$W(y, x) = W(x, y)^*$. Following \cite{W1}, we take the
potential entries and hopping terms  random and assume them to be
independent up to the Hermitian constraint, meaning that
\[ \left\{ V(x) \, \mid \, x \in \mathbb{Z}^d \right\}\,\,\bigcup\,\,
\left\{ W(x, y) \, \mid \, x,y\in\mathbb{Z}^d, \text{$x$ has even
sum of coordinates and }x\sim y\right\}\] are jointly independent.
We assume that either the distribution of each $V(x)$ is given by
the GOE, and that the matrices $W(x, y)$ are real (orthogonal case),
or that the distribution of each $V(x)$ is given by the GUE (unitary
case). Here, the probability density of the GOE (Gaussian Orthogonal
Ensemble) with respect to the Lebesgue measure on real symmetric
matrices is proportional to $\exp\left\{ - \frac{N}{4} \mathrm{tr}\,
V^2\right\}$, and the probability density of the GUE (Gaussian
Unitary Ensemble) with respect to the Lebesgue measure on Hermitian
matrices is proportional to $\exp\left\{ - \frac{N}{2} \mathrm{tr}\,
V^2\right\}$.

 Special cases of the model \eqref{eq:Z_d_model}
include the block Anderson model and the Wegner orbital model in
their orthogonal and unitary versions, given by
\begin{equation}\label{eq:def}\begin{split}
\text{\stackanchor{block}{Anderson}} & \begin{cases}
(H^{\text{bA},\mathbb{R}} \psi)(x) =
V^\text{GOE}(x) \psi(x) + g \sum_{y \sim x} (\psi(x) - \psi(y))~, \\
(H^{\text{bA},\mathbb{C}} \psi)(x) = V^\text{GUE}(x) \psi(x) + g
\sum_{y \sim x} (\psi(x) - \psi(y))~,
\end{cases} \\
\text{\stackanchor{Wegner}{orbital}} &\begin{cases}
(H^{\text{Weg},\mathbb{R}} \psi)(x) =
V^\text{GOE}(x) \psi(x) + g \sum_{y \sim x} W^\mathbb{R}(x,y) \psi(y)~, \\
(H^{\text{Weg},\mathbb{C}} \psi)(x) = V^\text{GUE}(x) \psi(x) + g
\sum_{y \sim x} W^\mathbb{C}(x,y) \psi(y)~,
\end{cases}
\end{split}\end{equation}
where $W^\mathbb{R}(x, y)$ has independent real Gaussian
$N_\mathbb{R}(0,1/N)$ entries,  $W^\mathbb{C}(x, y)$ has independent
complex Gaussian $N_\mathbb{C}(0,1/N)$ entries, $g > 0$ is a coupling
constant, and superscripts
indicate the symmetry class of the potential matrices. The block
Anderson model is a generalisation of the Anderson model \cite{And}
with Gaussian disorder (which is recovered when $N = 1$), whereas
the Wegner orbital model is invariant in distribution under local
gauge transformations, i.e.\ conjugation by $\mathcal{U}$ of the
form
\[ (\mathcal{U}\psi)(x) = \mathcal{U}(x) \psi(x)~, \quad \text{where} \quad \mathcal{U}(x) \in
\begin{cases}O_N~, &H^{\text{Weg},\mathbb{R}}\\
U_N~, &H^{\text{Weg},\mathbb{C}}
\end{cases}, \quad x \in \mathbb{Z}^d~. \]

Our results pertain to three topics: localisation at strong disorder
in arbitrary dimension (Theorem~\ref{thm:loc}), estimates on the
density of states (Wegner estimates, Theorem~\ref{thm:weg}), and on
the probabilities of multiple eigenvalues in a short interval
(Minami estimates, Theorem~\ref{thm:min}). The latter two results
are proved in greater generality, for deformed block-Gaussian
matrices, and are also applicable to Gaussian band matrices. The
common theme is the strive for the sharp dependence on the number
$N$ of orbitals (internal degrees of freedom). As an additional
application, we improve the upper bound from \cite{Sch1} on the
localisation length of one-dimensional Gaussian band matrices
(Theorem 4).

\paragraph{Strong disorder localisation}
The Anderson model in dimension $d\geq 3$ is conjectured to exhibit a spectral phase transition between a localisation (insulator) regime and a delocalisation (conductor) regime.
In particular, there should exist a threshold $g_0(d)$ such
that for $g < g_0(d)$ the spectrum is pure point, whereas
for $g > g_0(d)$ it has an absolutely continuous component.
So far only the localisation side of the transition has been established mathematically.
Two methods of proof are now available: the multi-scale analysis of Fr\"ohlich and Spencer \cite{FS} and the fractional moment method of Aizenman and Molchanov \cite{AM}.

A  phase transition similar to that of the Anderson model is
conjectured to occur, in dimension $d \geq 3$, for the orbital
models \eqref{eq:def}, with the threshold $g_0(d, N)$ depending on
the dimension and the number of orbitals. The first subject of the
current paper is the dependence of the threshold $g_0(d, N)$ on the
number of orbitals $N$. On the physical level of rigour, this
question was settled already in the original papers \cite{W1,SW,OW}.
The arguments provided there indicate that, for $d \geq 3$,
\begin{equation}\label{eq:conj}
g_0(d, N) \sim \left\{ C(d) \sqrt{N} \right\}^{-1}\quad\text{as
$N\to\infty$~.}
\end{equation}
Two heuristic arguments are discussed in  Section~\ref{s:rmks}.

\smallskip
Our first result, Theorem~\ref{thm:loc} below, is a mathematically
rigorous confirmation to one direction of  (\ref{eq:conj}), the
localisation side. The result is stated for the general model
\eqref{eq:Z_d_model}. When specialised to  the models
(\ref{eq:def}), it asserts that, for $g$ below the threshold
(\ref{eq:conj}), the matrix elements of the resolvent decay
exponentially in the distance from the diagonal. The formal
statement is in terms of finite-volume restrictions: denote by
\[ P_\Lambda: \ell_2(\mathbb{Z}^d \to \mathbb{C}^N) \to \ell_2(\Lambda \to \mathbb{C}^N)\]
 the coordinate projection to a finite volume
$\Lambda \subset \mathbb{Z}^d$ and, for an operator $H$ of the form
\eqref{eq:Z_d_model}, let
\begin{equation}\label{eq:H_Lambda_def}
H_\Lambda := P_\Lambda H P_\Lambda^*
\end{equation}
be the restriction of $H$ to $\Lambda$. Let $\|x - y\|_1$ be the
graph distance between  $x, y \in \mathbb{Z}^d$, let $\|v\|$ be the
$\ell_2$ norm of a vector $v\in\mathbb{C}^N$ and let
$\|W\|_{\operatorname{op}}$ be the operator norm of a matrix $W$.
\begin{thm}\label{thm:loc}
There exists a constant $C
> 0$ such that the following holds. Let $0 < s < 1$, let $H$ be as in \eqref{eq:Z_d_model} in either
the orthogonal case or the unitary case and suppose that
\begin{equation}
\label{eq:belthr} g_{\operatorname{eff}}:= \sup_{x,y}
\left\{ \mathbb{E}\|W(x,y)\|_{\operatorname{op}}^s\right\}^{\frac{1}{s}}< \left\{ \frac{1-s}{Cd}
\right\}^{\frac{1}{s}} \frac{1}{\sqrt{N}}~.
\end{equation}
Then for any finite $\Lambda \subset \mathbb{Z}^d$, $x, y \in
\Lambda$, $\lambda \in \mathbb{R}$, and $v \in \mathbb{C}^N$:
\begin{equation}\label{eq:expdecay}
\mathbb{E} \| (H_\Lambda- \lambda)^{-1} (x, y) \,  v \|^s \leq
\frac{CN^{s/2}}{1-s}  \left(\frac{Cd ({g_{\operatorname{eff}}} \sqrt{N})^s}{1-s}
\right)^{\| x -y \|_1}\|v\|^s~.\end{equation}
\end{thm}
In the left-hand side of (\ref{eq:expdecay}), we first take the
matrix inverse, then extract an $N \times N$ block, and then
multiply by a vector, or formally:
\[ (H_\Lambda- \lambda)^{-1} (x, y)\, v = P_{\{x\}} (H_\Lambda- \lambda)^{-1} P_{\{y\}}^* v~. \]
Also observe that the assumption (\ref{eq:belthr}) guarantees that
$\frac{Cd ({g_{\operatorname{eff}}}\sqrt{N})^s}{1-s}<1$, and thus the right-hand side
of (\ref{eq:expdecay}) indeed decays exponentially in $\|x-y\|_1$.

\medskip
\begin{rmk}\label{rmk:loc} Theorem~\ref{thm:loc} applies to the models (\ref{eq:def}) and yields the conclusion
(\ref{eq:expdecay}), where $g_{\operatorname{eff}}$ is replaced with
$g$. For $s=1 - \log^{-1}(d+2)$, the assumption (\ref{eq:belthr}) is
implied by
\begin{equation}\label{eq:thr} g < \left\{C d \log (d + 2) \sqrt{N}\right\}^{-1} \end{equation}
(where the constant may differ from that of (\ref{eq:belthr})).
\end{rmk}

\begin{rmk} Methods have been developed to pass from decay estimates
for the resolvent in finite volume to other signatures of Anderson
localisation, in particular, pure point spectrum and dynamical
localisation in infinite volume. We refer in particular to the
eigenfunction correlator method introduced by Aizenman \cite{A}; see
further \cite[Theorem~A.1]{ASFH}. Such methods can also be applied
in our setup. \end{rmk}

\medskip
The main feature of Theorem~\ref{thm:loc} is the dependence on the
number of orbitals, $N$, which, for the models \eqref{eq:def}, is
conjecturally sharp in dimension $d \ge 3$. For comparison,
localisation for
\[ g < \left\{ C(d) N^{3/2} \right\}^{-1}~,\]
follows from the  general theorems pertaining to variants of  the Anderson model (proved either by the method of \cite{FS}, or of \cite{AM}) and does not require the additional arguments of the current paper.

The asymptotics for growing $d$ in (\ref{eq:thr}) is the same as in
the corresponding result for the usual Anderson model; moreover, a
heuristic analysis of resonances suggests that $d \log d$ is the
true order of growth of the threshold (cf.\ Abou-Chacra, Anderson,
and Thouless \cite[(6.17)--(6.18)]{AAT} and a recent rigorous result
of Bapst \cite{Bapst} pertaining to the Anderson model on a tree).
Also, the arguments of \cite{Sch2} can be applied in the current
context, to express the constant $C$ in terms of the connectivity
constant of
 self-avoiding walk on $\mathbb{Z}^d$.

\paragraph{Wegner estimates}
Next we discuss Wegner estimates for a class of models, which
contains, in particular, the models (\ref{eq:Z_d_model}) and certain
Gaussian band matrices. For a given Hermitian matrix $H$ and an
interval $I\subset\mathbb{R}$ denote
\[
\mathcal{N} (H, I) = \#\{ \text{eigenvalues\, of\,  $H$\,  in\,  $I$}\}~.
\]
Also denote by $|I|$ the length of $I$.

Estimates on the density of states (cf.\ (\ref{eq:defdos}) and the
subsequent remark below) were first obtained, in the context of
Schr\"odin\-ger operators, by Wegner \cite{W3}, who proved the
following:

\smallskip
{\em Let $H_0$ be a  $k \times k$ Hermitian matrix,  and let $H =
H_0 + V$, where $V$ is a random diagonal matrix with entries
independently sampled from a bounded probability density $p$ on
$\mathbb{R}$. Then
\begin{equation}\label{eq:fweg}
\mathbb{E} \mathcal{N}(H, I) \leq \|p\|_\infty \, k  \, |I|~, \quad
\text{$I$ is an interval in $\mathbb{R}$}~.
\end{equation} }
The original motivation of Wegner was to rule
out the divergence of the density of states at the
mobility edge. Since then, estimates on the mean number
of eigenvalues in an interval, commonly referred to as Wegner estimates,
have found numerous applications in the mathematical
study of random operators (where they allow
to handle resonances).

The original estimate (\ref{eq:fweg}) can be applied to the
finite-volume restrictions of the models (\ref{eq:def}), where it
provides the sharp dependence on the volume and on the size of the
interval, but not on the number of orbitals. We prove a form of
(\ref{eq:fweg})
 tailored to the models
(\ref{eq:def}), with the sharp dependence on $N$. We formulate the
result in a more general form, which applies also to Gaussian band
matrices.

For positive integers $k, N_1, \ldots, N_k$, we consider a random square matrix of dimension $\sum_{j=1}^k N_j$ which has the form
\begin{equation}\label{eq:block_Gaussian_model}
  H = H_0 + \bigoplus_{j=1}^k V(j) = H_0 + \left( \begin{array}{c|c|c|c|c}
V(1) & 0 & 0 &  \cdots & 0\\\hline
0 & V(2) & 0 & \cdots & 0\\\hline
0 & 0 & V(3) & \cdots & 0\\\hline
&&&&\\\hline
0 & 0 & 0 &  \cdots & V(k)
\end{array}\right)
\end{equation}
in which $H_0$ is deterministic and the matrices $(V(j))$ are random and independent, with $V(j)$ of size $N_j\times N_j$. We
assume that either the distribution of each $V(j)$ is given by the GOE,
and that $H_0$ is real symmetric (orthogonal case), or that
the distribution of each $V(j)$ is given by the GUE and $H_0$ is Hermitian (unitary case). We refer to
 matrices thus defined as deformed block-Gaussian matrices.

\begin{thm}\label{thm:weg} There exists a constant $C>0$ such that the following holds. Let $H$ be a deformed block-Gaussian matrix as in \eqref{eq:block_Gaussian_model}, in either the orthogonal case or the unitary case. Then,
for any interval $I \subset \mathbb{R}$,
\begin{equation}
\mathbb{E} \mathcal{N}(H, I)  \leq C \sum_{j=1}^k N_j \, |I|~.
\end{equation}
\end{thm}
\noindent The unitary case of Theorem~\ref{thm:weg} was also
recently proved by Pchelin \cite{Pch}. Pchelin relies on a
single-block estimate ($k=1$, cf.\ Proposition~\ref{prop:vec} and
(\ref{eq:dos}) below) which he proves in the unitary case. Possibly,
his argument can rely additionally on the orthogonal case of the
single-block estimate, as proved in \cite{APSSS}, and yield an
alternative proof of Theorem~\ref{thm:weg} in full generality.

Our proofs of Theorem~\ref{thm:weg} and Theorem~\ref{thm:min} below
rely on a representation formula for $\mathcal{N}(H,I)$ in terms of
similar quantities for single blocks. This formula, presented in
Proposition~\ref{l:formula} (see also Remark~\ref{rmk:normal}), may
possibly be of use elsewhere.

One can obtain complementary bounds to Theorem~\ref{thm:weg}; see
Section~\ref{s:rmks}.

\medskip\noindent{\em Application \# 1: orbital model.}
Going back to the orbital operators (\ref{eq:Z_d_model}), the
theorem applies to the restriction of each of them to a finite
volume. For integers $L\ge 0$ and $d\ge 1$ we write
\begin{equation*}
  \Lambda_{L}^d := \{-L,-L+1,\ldots, L\}^d.
\end{equation*}

\begin{cor}\label{cor:orbital_Wegner} There exists a constant $C>0$ such that the following holds. Let $H$ be as in \eqref{eq:Z_d_model}
in either the orthogonal case or the unitary case, and let
$H_\Lambda$ be the restriction of $H$ to a finite volume
$\Lambda\subset\mathbb{Z}^d$ as in \eqref{eq:H_Lambda_def}. Then
\[ \mathbb{E} \mathcal{N}(H_\Lambda, I) \leq C N |\Lambda| |I| \]
for any interval $I \subset \mathbb{R}$. In particular, if the
limiting measure
\begin{equation}\label{eq:defdos} \rho_H(\cdot) = \lim_{L \to \infty}  \frac{\mathbb{E} \, \mathcal{N}(H_{\Lambda_{L}^d}, \cdot)}{N(2L+1)^d}
\end{equation}
exists, then it has a density (called the density of states of $H$)
which is bounded uniformly in $N$.
\end{cor}
We remark that according to general results pertaining to metrically
transitive [= ergodic] operators, see Pastur and Figotin \cite{FP}
or Aizenman and Warzel \cite{AW}, the limiting measure in
\eqref{eq:defdos} exists for the models \eqref{eq:def} and, more
generally, whenever the distribution of the hopping terms $W(x,y)$
depends only on $x-y$.

\begin{proof}[Proof of Corollary~\ref{cor:orbital_Wegner}] Condition on the hopping terms $W(x,y)$ and apply Theorem~\ref{thm:weg} with $k=|\Lambda|$ and all $N_j$ equal to $N$. \end{proof}

\medskip\noindent
Let us briefly discuss related previous results. Constantinescu,
Felder, Gaw{\k{e}}dzki, and Kupiainen \cite{CFGK} derived an
integral representation of the density of states for a class of
locally gauge-invariant operators including $H^\text{Weg}$. Using
this representation, they proved, for a specific model slightly
outside the class (\ref{eq:Z_d_model}), that the density of states
is analytic, uniformly in $N$, in a certain range of parameters. In
the case of $d=1$, further results pertaining to the density of
states were obtained by Constantinescu \cite{C} using supersymmetric
transfer matrices.

The integrated density of states (the cumulative distribution
function of the measure $\rho_H$ from (\ref{eq:defdos})) was studied
by Khorunzhiy and Pastur \cite{KhP}, who established, for a wide
class of orbital models, an asymptotic expansion in inverse powers
of $N$; see further Pastur \cite{P2} and the book \cite[\S
17.3]{PShch} of Pastur and Shcherbina.

\medskip \noindent
{\em Application \# 2: Gaussian band
matrices.} We proceed to define a class of Gaussian random matrices
to which the results can be applied, and which contains the class of Gaussian
band matrices.
We say
that a random variable is complex Gaussian if its real and imaginary
parts are independent real Gaussian random variables with equal
variance.
\begin{defin}\label{defin:band} Let $L\ge 0$, $d\ge 1$ be integers and let $\psi:\mathbb{Z}^d \to [0,\infty)$
satisfy $\psi(-r) = \psi(r)$. A Gaussian random
matrix $H_L$ with domain $\Lambda_{L}^d$ and shape function $\psi$
is an Hermitian $(2L+1)^d \times (2L+1)^d$ random matrix, whose rows
and columns are indexed by the elements of $\Lambda_{L}^d$, having
the form
\begin{equation*}
  H_L = \frac{X_L + X_L^*}{\sqrt{2}},
\end{equation*}
where the entries of the matrix $X_L$ are either independent real
Gaussian (orthogonal case) or independent complex Gaussian (unitary
case), having zero mean and satisfying
\begin{equation*}
  \mathbb{E}|X_L(x,y)|^2 = \psi(x-y),\quad x,y\in\Lambda_{L}^d.
\end{equation*}
\end{defin}
We remark that an equivalent way to specify the covariance structure
of $H_L$ in the above definition is via the formula
\begin{equation}\label{eq:defband}
\mathbb{E} H_L(x, y) \overline{H_L(x', y')}  = \psi(x-y) \times
\begin{cases}
 \mathbbm{1}_{x=x', y=y'} + \mathbbm{1}_{x=y', y=x'}~,
&\text{orthogonal case} \\
 \mathbbm{1}_{x=x', y=y'}~,
&\text{unitary case}~.\end{cases}
\end{equation}
\begin{rmk}\label{rmk:GOE_GUE_distribution}
  We note for later use that, in our normalization, an $N\times N$
  random GOE (GUE) matrix has the same distribution as $\frac{X +
  X^*}{\sqrt{2N}}$ where the entries of the matrix $X$ are independent real (complex)
Gaussian with zero mean and with $\mathbb{E}|X(x,y)|^2=1$ for all $x,y$.
\end{rmk}

Theorem~\ref{thm:weg} implies a Wegner estimate for the Gaussian random matrices
thus defined. We write $\|v\|_\infty$ for the
$\ell^\infty$ norm of a vector $v$.
\begin{cor}\label{cor:metrically_transitive_Wegner}
There exists a constant $C>0$ such that the following holds. Let
$H_L$ be a Gaussian random matrix with domain
$\Lambda_{L}^d$ and shape function $\psi$ in either the orthogonal
case or the unitary case. Suppose that
\begin{equation}\label{eq:metrically_transitive_boundedness}
  \psi(r)\ge \frac{1}{(W+1)^d}\quad\text{when}\quad \|r\|_\infty\le 2
  \min(W, L)
\end{equation}
for some integer $W$ satisfying $0\le W\le 2L$. Then for any
interval $I \subset \mathbb{R}$,
\begin{equation}\label{eq:metrically_transitive_Wegner}
 \mathbb{E} \mathcal{N}(H_L, I)  \leq C(2L+1)^d |I|~.
\end{equation}
In particular, assuming \eqref{eq:metrically_transitive_boundedness}
holds for some integer $W\ge 0$, the measure
\[
\rho(\cdot) = \lim_{L\to
\infty}\frac{\mathbb{E}\mathcal{N}(H_L,\cdot)}{(2L+1)^d}
\]
has a density, the density of states, which is uniformly bounded by
$C$.
\end{cor}

The corollary is particularly interesting in the case when $W$ is a
large parameter, $L \gg W$, and $\psi(r)$ is small for $\|r\|_\infty
\gg W$; in this case $H_L$ is informally called a Gaussian band
matrix of bandwidth $W$. One way to construct such matrices is to
choose, slightly modifying the  definition used by Erd\H{o}s and
Knowles \cite{EK2}, the shape function $\psi$ of the form
\begin{equation}\label{eq:ekband}
  \psi(r) = \frac{\phi(\frac{r}{W})}{W^d},\quad r\in\mathbb{Z}^d
\end{equation}
for an almost everywhere continuous function $\phi:\mathbb{R}^d \to [0,\infty)$
satisfying $\phi(-r) = \phi(r)$ and $0<\int \phi(r) dr<\infty$.
If $H_L$ is constructed in this way, and if
\begin{equation}\label{eq:bdd0}
\phi(\rho) \geq \delta  \quad \text{for} \quad \| \rho \|_\infty
\leq \epsilon~
\end{equation}
with some $0<\epsilon\le \frac{4L}{W}$ and
$\delta>0$, then, for any interval $I \subset \mathbb{R}$,
\begin{equation}\label{eq:band_matrix_Wegner}
\mathbb{E} \mathcal{N}(H_L, I) \leq
K \, (2L+1)^d |I|~,
\end{equation}
where $K =C
\sqrt{\frac{1}{\delta}\left(\frac{2}{\epsilon}\right)^d}$. This
follows from Corollary~\ref{cor:metrically_transitive_Wegner}
applied to to the matrix
$\sqrt{\frac{1}{\delta}\left(\frac{2}{\epsilon}\right)^d} H_L$ with
$\lfloor\frac{\epsilon W}{2}\rfloor$ in place of $W$.

Another example of Gaussian band matrices, in which Corollary
\ref{cor:metrically_transitive_Wegner} can be applied to deduce
(\ref{eq:band_matrix_Wegner}) with a constant $K$ independent of
$W$, is given by
\begin{equation}\label{eq:susy}\begin{split}
&\mathbb{E} H_L(x, y) \overline{H_L(x', y')}  \\
&= (-W^2 \Delta + \mathbbm{1})^{-1}(x, y)  \times
\begin{cases}
 \mathbbm{1}_{x=x', y=y'} + \mathbbm{1}_{x=y', y=x'}~,
&\text{orthogonal case} \\
 \mathbbm{1}_{x=x', y=y'}~,
&\text{unitary case}~,\end{cases}
\end{split}
\end{equation}
where $\Delta$ is the discrete Laplacian on $\mathbb{Z}^d$, $d\ge
1$.

This example was studied by Disertori, Pinson, and Spencer
\cite{DPS}, who proved an estimate of the form
 (\ref{eq:band_matrix_Wegner}) for the unitary case in dimension $d=3$. Very recently,
a parallel result for $d = 1$ was proved by M.\ and T.\ Shcherbina
\cite{ShchX2}, and for $d = 2$ --- by Disertori and Lager \cite{DL}.

To the best of our knowledge, these are the only previously known
estimates of the form (\ref{eq:band_matrix_Wegner}) for any kind of band matrices
which are valid for arbitrarily short intervals $I$ uniformly in $W$;  see further
\cite[\S 3]{Sp_band} for a discussion of the problem.
We remark that the methods of \cite{DPS,DL} and \cite{ShchX2} allow to
go beyond a uniform bound on the density of states, and provide a differentiable
asymptotic expansion for it in powers of $W^{-2}$. On the other hand, these methods make essential use of the particular structure (\ref{eq:susy}).

In a generality similar to that of Definition~\ref{defin:band},
Bogachev, Molchanov, and Pastur \cite{BMP} and Khorunzhiy,
Molchanov, and Pastur \cite{KMP} found the limit of
$\mathcal{N}(H_L, I)/(2L+1)^d$ (with or without the expectation)
for a fixed interval $I$ as $W, L \to \infty$; this limit
is bounded by a constant times the length of $I$. The results of
Erd\H{o}s, Yau, and Yin \cite{EYY} (and, in a slightly different
setting, of \cite{band2}) yield an estimate of the form
(\ref{eq:band_matrix_Wegner}) for intervals $I$ of length $|I| \geq
W^{-1+\epsilon}$.

\begin{proof}[Proof of Corollary~\ref{cor:metrically_transitive_Wegner}]
Using the assumption that $0\le W\le 2L$ we may partition
$\{-L,-L+1,\ldots, L\}$ into disjoint discrete intervals $I_j$,
$1\le j\le\ell$, satisfying $W+1\le |I_j|\le 2W+1$ for all $j$ (if
$W\ge L$ then the partition necessarily has $\ell = 1$ and $|I_1| =
2L+1$). Correspondingly, write
\begin{equation}\label{eq:index_set_decomposition}
\Lambda_{L}^d = \biguplus_{j=1}^{\ell^d} B_j
\end{equation}
where the $B_j$ are all Cartesian products of the form $J_1\times
J_2\times\cdots\times J_d$ where each $J_i$ is one of the intervals
$I_j$.

Now consider the matrix $H_L$ as a block matrix, where the partition
of the index set $\Lambda_{L}^d$ into blocks is given by
\eqref{eq:index_set_decomposition}. The assumption
\eqref{eq:metrically_transitive_boundedness}, the fact that the
entries of $H_L$ are Gaussian and the observation in
Remark~\ref{rmk:GOE_GUE_distribution} allow us to write
\begin{equation}
  H_L = H_L^0 + V_L
\end{equation}
where $V_L$ is a block-diagonal matrix, with the diagonal blocks
distributed as GOE in the orthogonal case and as GUE in the unitary
case, and where $H_L^0$ is an Hermitian matrix, independent of
$V_L$, with jointly Gaussian entries which are real in the
orthogonal case. Thus, conditioning on $H_L^0$, the estimate
\eqref{eq:metrically_transitive_Wegner} follows from
Theorem~\ref{thm:weg} applied with $k = \ell^d$ and $N_j = |B_j|$.
\end{proof}

\paragraph{Minami estimates}

In the same setting as (\ref{eq:fweg}), Minami established
\cite{Minami} the bound:
\begin{equation}\label{eq:fminami}
\mathbb{E} \mathcal{N}(H, I) (\mathcal{N}(H, I) - 1) \leq (C \|p\|_\infty |\Lambda| \, |I|)^2~.
\end{equation}
The bound (\ref{eq:fminami}) rules out attraction between
eigenvalues in the local regime; it is a key step in Minami's proof
of Poisson statistics for the Anderson model in the regime of
Anderson localisation. Subsequently, additional proofs and
generalisations of (\ref{eq:fminami}) were found, among which we
mention the argument of Combes, Germinet and Klein \cite{CGK}.

The next result is a counterpart of (\ref{eq:fminami}) in our block setting. As in
Theorem~\ref{thm:weg}, the central feature is the dependence on the sizes of the blocks.

\begin{thm}\label{thm:min} There exists $C>0$ such that the following holds. Let $H$ be a deformed
block-Gaussian matrix as in \eqref{eq:block_Gaussian_model}, in either the orthogonal case or the unitary case. Then,
for any integer $m\ge 1$ and interval $I \subset \mathbb{R}$,
\begin{equation}\label{eq:min_k}
\mathbb{E} \prod_{\ell=0}^{m-1} (\mathcal{N}(H, I)-\ell) \leq \left(C
\sum_{j=1}^k N_j \, |I|\right)^m~,
\end{equation}
and, consequently,
\[ \mathbb{P} \left\{ \mathcal{N}(H, I) \geq m \right\} \leq \frac{1}{m!} \left(C \sum_{j=1}^k N_j \, |I|\right)^m~. \]
\end{thm}
The case $m=1$ in the theorem is the Wegner estimate discussed in
Theorem~\ref{thm:weg} whereas the cases $m\ge 2$ are Minami-type
estimates.

\paragraph{Localisation for one-dimensional band matrices}  Band matrices in one dimension ($d=1$) have been studied extensively in the physics literature as a simple model in which
the quantum dynamics exhibits crossover from quantum diffusion to
localisation, see \cite{Casati:1990p3447, Casati:1993p35, FM}. Based
on those works, the following crossover is expected: considering
band matrices of dimension $L$ and bandwidth $W$, when $W \ll
\sqrt{L}$, each eigenvector has appreciable overlap with a
vanishingly small fraction of the standard basis vectors in the
large $L$ limit, whereas for $W \gg \sqrt{L}$ a typical eigenvector
has overlap of order $1/\sqrt{L}$ with most standard basis vectors.
A related conjecture states that the $i,j$-entry of the resolvent
should decay as $\exp(- C |i-j| W^{-2})$ for $W \ll \sqrt{L}$.

In \cite{Sch1}, one of us studied the localisation side of this
problem. In that paper it was shown that certain ensembles of random
matrices whose entries vanish outside a band of width $W$ around the
diagonal satisfy a localisation condition  in the limit that the
size of the matrix $L$  tends to infinity such that $W^{8}/L
\rightarrow 0$.  For Gaussian band matrices, our present work
settles  \cite[Problem 2]{Sch1} in the positive, thereby allowing to improve the
result there slightly by replacing the exponent $8$ with the
exponent $7$ (which is still a bit away from the expected optimal
exponent $2$).
\begin{thm}\label{thm:resolvent}  Let $H_L$ be a Gaussian random matrix with domain $\Lambda_L = \{-L,-L+1,\ldots,
L\}$ and shape function $\psi$ as in Definition \ref{defin:band} in
either the orthogonal case or the unitary case. Let $W$ be an
integer dividing $2L+1$ and suppose that $\psi$ is the sharp cutoff
function
\begin{equation}\label{eq:sharp}\psi(r) = \begin{cases} \frac{1}{W} & |x|< W \\ 0 & |x| \ge W \end{cases}~.\end{equation}
Then, given $\rho >0$ and $s  \in (0,1)$ there are $A < \infty$
and $\alpha > 0$ such that
\begin{equation}\label{eq:resolvloc}
  \mathbb{E}\left ( \left | (H_L - \lambda)^{-1}(i,j)\right |^s \right ) \ \le \ A  W^{\frac{s}{2}} e^{-\alpha \frac{|i-j|}{W^{7}}}
\end{equation}
for all $\lambda \in [-\rho,\rho]$ and all $i,j\in\{-L,\ldots, L\}$.
\end{thm}

\begin{rmk} The theorem implies (using the resolvent
identity) that a similar estimate holds
without the assumption $2L+1 \equiv 0 \mod W$. 
\end{rmk}

\section{Proof of the theorems}

In this section we prove the main results of the paper. We use the
following result from \cite{APSSS}, where the  object of study was
the regularizing effect of adding a Gaussian random matrix to a
given deterministic matrix. The GUE case of (\ref{eq:dos}) was also
proved by Pchelin \cite{Pch}.
\begin{prop}\cite[Theorem 1 and Remark 2.2]{APSSS}\label{prop:vec}
\hfill
\vspace{-.2cm}
\begin{tabbing}
If either:  \= $A$ is
  an $N\times N$ real symmetric matrix, \=$v\in\mathbb{R}^N$, \=and $V$ is sampled from  GOE,  \\
or: \> $A$ is
  an $N\times N$ Hermitian matrix,  \>$v\in\mathbb{C}^N$, \>and  $V$ is sampled from  GUE,  \\
\end{tabbing}
\vspace{-.7cm}
then  the matrix  $A+V$ satisfies the bounds:
\begin{align}\label{eq:vec1}
&\mathbb{P} \left\{ \| (A+V)^{-1} v \| \geq t \sqrt{N} \|v\|\right\} \leq \frac{C}{t} \, , \quad t \geq 1,\\
\label{eq:dos}
&\mathbb{E} \,  \mathcal{N}(A + V, I) \leq \   C   N|I|~, \quad \text{$I$ is an interval in $\mathbb{R}$}
\end{align}
with a constant $C<\infty $ which is uniform in  $N$, $A$, and $v$.
Moreover, the following stronger version of \eqref{eq:vec1} holds:
almost surely,
\begin{equation}\label{eq:cond_vec}
\mathbb{P} \left\{ \| (A+V)^{-1} v \| \geq t \sqrt{N} \|v\| \;\Big|\; \hat{P}_{v^\perp} V \hat{P}_{v^\perp}^*\right\} \leq \frac{C}{t} \, ,
\quad t \geq 1~,\\
\end{equation}
where $\hat{P}_{v^\perp}: \mathbb{C}^N \to \mathbb{C}^{N} /
\mathbb{C}v \simeq v^\perp$ is the canonical projection, and
$v^\perp$ is the orthogonal complement to $v$.
\end{prop}
In the setting of Proposition~\ref{prop:vec} we note the following
consequence of (\ref{eq:vec1}),
\begin{equation}\label{eq:vec}
\mathbb{E} \| (A+V)^{-1} v \|^s \leq \frac{C_0 N^{\frac{s}{2}}}{1-s} \|v\|^s~, \quad 0 < s < 1~,
\end{equation}
where $C_0$ is an absolute constant (uniform in all the parameters
of the problem).

\begin{rmk}
Our proofs of Theorem~\ref{thm:loc} (localisation), Theorem~\ref{thm:weg} (Wegner-type estimate) and
Theorem~\ref{thm:min} (Minami-type estimate) rely on the Gaussian
structure of the underlying random matrix ensembles only through
Proposition~\ref{prop:vec} (and the simple observation of
Remark~\ref{rmk:not_eigenvalue} below). Thus, if an extension of the
proposition to other random matrix ensembles is found, corresponding
extensions of our theorems will follow. 
\end{rmk}

\begin{rmk}\label{rmk:not_eigenvalue}
  For the random matrix models discussed in our theorems, any given $\lambda\in\mathbb{R}$ is almost surely
  not an eigenvalue. This follows, for instance, from the following observation: if
  $H$ is a random matrix satisfying that for any $\mu\in\mathbb{R}$, the
  distribution of $H$ and the distribution of $H-\mu$ are
  mutually absolutely continuous then any given $\lambda\in\mathbb{R}$ is
  almost surely not an eigenvalue of $H$.
\end{rmk}

\begin{proof}[Proof of Theorem~\ref{thm:loc}]
Denote by $G_\lambda[\tilde{H}] = (\tilde{H} - \lambda)^{-1}$ the
resolvent of an operator $\tilde{H}$. For $\tilde{x} \in
\tilde{\Lambda} \subset \Lambda$, let $\mathcal{F}_{\tilde{\Lambda},
\tilde{x}}$ be the $\sigma$-algebra generated by all
$H_{\tilde{\Lambda}}(w, w')$, where $w, w' \in \tilde{\Lambda}$ and
$(w, w') \neq (\tilde{x}, \tilde{x})$.

Observe the following corollary of the Schur--Banachiewicz formula
for block matrix inversion: for any $\tilde{x} \in \tilde{\Lambda}
\subset \Lambda$,
\begin{equation}\label{eq:sbeq}
G_\lambda[H_{\tilde\Lambda}](\tilde{x}, \tilde{x}) = (V(\tilde{x}) - \lambda - \Sigma)^{-1}~,
\end{equation}
where $\Sigma$ is measurable with respect to
$\mathcal{F}_{\tilde{\Lambda}, \tilde{x}}$. Consequently, by
(\ref{eq:vec}), almost surely,
\begin{equation}\label{eq:sbeqcor}
\mathbb{E} \left[ \left\| G_\lambda[H_{\tilde\Lambda}](\tilde{x},
\tilde{x}) \tilde{v} \right\|^s \, \mid \,
\mathcal{F}_{\tilde{\Lambda}, \tilde{x}}\right] \leq \frac{C_0
N^{\frac{s}{2}}}{1-s} \mathbb{E} \|\tilde{v}\|^s
\end{equation}
whenever $\tilde{v}$ is a random vector which is measurable with
respect to $\mathcal{F}_{\tilde{\Lambda}, \tilde{x}}$.

Next, we use the following representation of $G_\lambda[H_\Lambda](x, y)$:
\begin{equation}\label{eq:rwe}
\begin{split}
&G_\lambda[H_\Lambda](x, y) = \sum_{k \geq \|x-y\|_1} (-1)^k \sum_{\pi \in \Pi_k(x, y)}
G_{\lambda}[H_\Lambda](\pi_0, \pi_0)   W(\pi_0, \pi_1)
G_{\lambda}[H_{\Lambda\setminus \{\pi_0\}}] (\pi_1, \pi_1) \\
&\quad W(\pi_1, \pi_2) G_{\lambda}[H_{\Lambda\setminus \{\pi_0, \pi_1\}}] (\pi_2, \pi_2) \cdots W(\pi_{k-1}, \pi_k) G_{\lambda}[H_{\Lambda\setminus\{\pi_0,\pi_1,\cdots,\pi_{k-1}\}}](\pi_k, \pi_k)~,
\end{split}
\end{equation}
where for $y = x$ the right-hand side is interpreted as $G_{\lambda}[H_\Lambda](x, x)$, and for
$y \neq x$ the collection $\Pi_k(x, y)$ includes all tuples of pairwise distinct vertices
$\pi_0, \pi_1, \cdots, \pi_k \in \Lambda $ such that $x = \pi_0 \sim \pi_1 \sim \pi_2 \cdots \sim \pi_k = y$.

Indeed, the representation is tautological for $y=x$, and for $y\neq x$ it follows by iterating the equality
\[ G_\lambda[H_\Lambda](x, y) = - \sum_{\pi_1 \sim x} G_\lambda[H_{\Lambda}](x, x) W(x, \pi_1) G_\lambda[H_{\Lambda \setminus \{x\}}] (\pi_1, y)~. \]
The latter is in turn a corollary of the second resolvent identity applied to  the operators $H_\Lambda - \lambda$ and $H_\Lambda^x - \lambda$, where $H_\Lambda^x$ is obtained from  $H_\Lambda$ by setting
the blocks $W(x, x')$ and $W(x', x)$ to $0$ for all $x' \sim x$.

Now we turn to the proof of the theorem. We derive from (\ref{eq:rwe}) using the triangle inequality
and $|a+b|^s \leq |a|^s + |b|^s$ that
\begin{equation}\label{eq:rwe'}
\begin{split}
&\|G_\lambda[H_\Lambda](x, y)v\|^s \leq \sum_{k \geq \|x-y\|_1} \sum_{\pi \in \Pi_k(x, y)}
\Big{\|}G_{\lambda}[H_\Lambda](\pi_0, \pi_0)  W(\pi_0, \pi_1)  G_{\lambda}[H_{\Lambda\setminus \{\pi_0\}}](\pi_1, \pi_1) \\
&\quad W(\pi_1, \pi_2)  G_{\lambda}[H_{\Lambda\setminus \{\pi_0, \pi_1\}}] (\pi_2, \pi_2)\cdots W(\pi_{k-1}, \pi_k) G_{\lambda}[H_{\Lambda\setminus\{\pi_0,\pi_1,\cdots,\pi_{k-1}\}}](\pi_k, \pi_k) v \Big{\|}^s~.
\end{split}
\end{equation}
To bound the expectation of a term in (\ref{eq:rwe'}), we repeatedly use (\ref{eq:sbeqcor}) and
the inequality
\begin{equation} \mathbb{E} \|W(\tilde{x}, \tilde{x}')\|_{\operatorname{op}}^s \leq g_{\operatorname{eff}}^s \end{equation}
from (\ref{eq:belthr}). We obtain for the term in (\ref{eq:rwe'}) corresponding to a single $\pi \in \Pi_k(x, y)$:
\begin{equation*}
\mathbb{E} \left\| G_{\lambda}[H_\Lambda](\pi_0, \pi_0)  W(\pi_0, \pi_1) \cdots G_{\lambda}[H_{\Lambda\setminus\{\pi_0,\pi_1,\cdots,\pi_{k-1}\}}](\pi_k, \pi_k) v \right\|^s \leq \left( \frac{C_0 N^{\frac{s}{2}}}{1-s}\right)^{k+1} g_{\operatorname{eff}}^{sk} \, \, \|v \|^s~.
\end{equation*}
The cardinality of $\Pi_k(x, y)$ does not exceed $(2d)^k$. Therefore
\[ \begin{split}
\mathbb{E} \|G_\lambda[H_\Lambda](x, y)v\|^s
&\leq \sum_{k \geq \|x-y\|_1} \left( \frac{C_0 N^{\frac{s}{2}}}{1-s}\right)^{k+1} \left( 2d g_{\operatorname{eff}}^{s} \right)^{k} \, \, \|v \|^s \\
 &\leq 2 \left( \frac{C_0N^{\frac{s}{2}}}{1-s}\right)^{\|x-y\|_1+1} \left( 2d g_{\operatorname{eff}}^{s} \right)^{\|x-y\|_1} \, \, \|v \|^s
 \end{split} \]
 whenever
 \[  \frac{4C_0 d g_{\operatorname{eff}}^{s} N^{\frac{s}{2}}}{1-s} \leq 1~.\]
This is what is claimed in the statement of the theorem, for $C = 4 C_0$.
\end{proof}

\medskip\noindent
The proofs of Theorems~\ref{thm:weg} and \ref{thm:min} are preceded by the following proposition, the purpose of which  is to write $\mathcal{N}(H, I)$ as a linear expression involving terms of the form $\mathcal{N}(V + A, J)$, where $V$ is a random matrix sampled from the GOE (or GUE), $A$ is a symmetric (or Hermitian) matrix independent of $V$ and $J$ is an interval in $\mathbb{R}$.
\begin{prop}\label{l:formula} Let $H$ have the form $H = H_0 + \oplus_{j=1}^k V(j)$ (as in \eqref{eq:block_Gaussian_model}) in which $H_0$ and all $V(j)$ are (deterministic) Hermitian matrices. Then for any interval $I$ the  endpoints of which are not eigenvalues of $H$,
\begin{equation*}
\mathcal{N}(H, I) = \lim_{\eta \to +0} \sum_{j=1}^k\int_I\frac{d\lambda}{\pi} \int_{-\infty}^\infty \frac{dt}{\pi(1+t^2)}
\int_0^\infty  \frac{2\eta \xi\, d\xi}{(\xi^2 + \eta^2)^2}\,\mathcal{N}\Big(V(j) + A(j,\lambda,\eta,t), (-\xi, \xi)\Big)~,
\end{equation*}
where $A(j,\lambda,\eta,t)$ is an Hermitian matrix determined by
$H_0$ and $(V(\ell))_{\ell\neq j}$ (that is, every matrix element of
$A$ is a Borel-measurable function of these variables and $\lambda,
\eta$ and $t$). In addition, if $H_0$ and all $V(j)$ are real, then
the matrices $A$ are real as well.\end{prop}

The exact definition of the matrices $A(j,\lambda, \eta, t)$ is
given in \eqref{eq:A_j_lambda_eta_t_def_1} and
\eqref{eq:A_j_lambda_eta_t_def_2}.
\begin{rmk}\label{rmk:normal} Each of the integrals
\[ \int_I\frac{d\lambda}{|I|}~, \quad  \int_{-\infty}^\infty \frac{dt}{\pi(1+t^2)}~, \quad
\int_0^\infty  \frac{4\eta \xi^2\, d\xi}{\pi (\xi^2 + \eta^2)^2}\]
equals one. This leads us to introduce the following notation:
\[ \operatorname{Ave}_{\lambda,t,\xi}^{ \eta} \Phi(\lambda, t,\xi; \eta )
= \int_I\frac{d\lambda}{|I|} \int_{-\infty}^\infty \frac{dt}{\pi(1+t^2)} \int_0^\infty  \frac{4\eta \xi^2\, d\xi}{\pi (\xi^2 + \eta^2)^2}
\Phi(\lambda, t,\xi; \eta)~. \]
With this notation and assuming  $0 < |I| < \infty$, the conclusion of Proposition~\ref{l:formula} takes the form:
\begin{equation}\label{eq:formula} \frac{1}{|I|} \mathcal{N}(H, I) =   \lim_{\eta \to +0} \sum_{j=1}^k \operatorname{Ave}_{\lambda,t,\xi}^{\eta}\,
\frac{1}{2\xi} \mathcal{N}\Big(V(j) + A(j,\lambda,\eta,t), (-\xi, \xi)\Big)~.
\end{equation}
\end{rmk}

We prove Theorem~\ref{thm:weg} and Theorem~\ref{thm:min} using
Proposition~\ref{l:formula}, and defer the proof of the proposition
to the next section.

\begin{proof}[Proof of Theorem~\ref{thm:weg}]
In the setting of the theorem, the endpoints of any fixed interval
$I$ are almost surely not eigenvalues of $H$ (see
Remark~\ref{rmk:not_eigenvalue}). Therefore
Proposition~\ref{l:formula} is applicable  almost surely, and
(\ref{eq:formula}) yields:
\begin{equation*}\begin{split}
\frac{1}{|I|}\mathbb{E} \mathcal{N}(H, I) &= \mathbb{E}  \lim_{\eta \to +0} \sum_{j=1}^k \operatorname{Ave}_{\lambda,t,\xi}^{\eta}\,
\frac{1}{2\xi} \mathcal{N}\Big(V(j) + A(j,\lambda,\eta,t), (-\xi, \xi)\Big)\\
&
\leq \lim_{\eta \to +0} \sum_{j=1}^k \operatorname{Ave}_{\lambda,t,\xi}^{\eta}\,
\frac{1}{2\xi} \mathbb{E}  \mathcal{N}\Big(V(j) + A(j,\lambda,\eta,t), (-\xi, \xi)\Big)\\
&\leq
\lim_{\eta \to +0}\sum_{j=1}^k  \operatorname{Ave}_{\lambda,t,\xi}^{\eta}\,
C N_j  = C\sum_{j=1}^k N_j~.
\end{split}\end{equation*}
The first inequality follows from the Fatou lemma and the second one is an application of the single-block
bound (\ref{eq:dos}) with $|I| = 2\xi$.
\end{proof}

The proof of  Theorem~\ref{thm:min} also uses
formula \eqref{eq:formula} as a starting point, and proceeds following
arguments similar to those used in the proof of
\cite[Theorem~2]{APSSS} (which is the $k=1$ case of Theorem~\ref{thm:min}); these arguments are, in turn, inspired by
the work of Combes, Germinet and Klein \cite{CGK}.
We start with the following simple lemma (see e.g.\ \cite[Lemma 3.1 and (3.6)]{APSSS}  for a slightly stronger version featuring the Frobenius norm in place of the operator norm).
\begin{lemma}\label{l:qf1} Let $A$ be an $N\times N$ deterministic Hermitian matrix.
\vspace{-.2cm}
\begin{tabbing}
If either:  \= $v$ is uniformly distributed on the sphere $\mathbb{S}^{N-1}_{\mathbb{R}} = \{w\in\mathbb{R}^N\colon\|w\|=1\}$ \= and $A$ is
  real,\\
or: \> $v$ is uniformly distributed on the  sphere $\mathbb{S}^{N-1}_{\mathbb{C}} = \{w \in\mathbb{C}^N\colon\|w \|=1\}$,\> \\
\end{tabbing}
\vspace{-.7cm}
then
\begin{equation*}
  \mathbb{P} \left\{ \|A v\| \leq \frac{\epsilon}{\sqrt{N}}\|A\|_{\operatorname{op}}
\right\} \leq 5 \epsilon, \quad \epsilon>0.
\end{equation*}
\end{lemma}
Consequently, in the same setting, for any non-negative random variable $X$ which is independent of $v$,
\begin{equation}\label{eq:expectation_via_uniform_vector}
    \mathbb{E}\, X \le 2\mathbb{E}\left[X\cdot \mathbbm{1}\left\{\|A^{-1}v\|\ge \frac{\|A^{-1}\|_{\operatorname{op}}}{10\sqrt{N}}\right\}\right],
  \end{equation}
where $\mathbbm{1}\{\Omega\}$ is the indicator of an event $\Omega$.
\begin{proof}[Proof of Theorem~\ref{thm:min}] It suffices to prove the theorem for
intervals $I$ of length $0 < |I| < \infty$, therefore we tacitly impose this assumption
on all intervals which appear in this proof. The argument is by induction on $m$. Let $C_1 = 10 \, C$, where $C$ is the greater among
the constants in Theorem~\ref{thm:weg} and Proposition~\ref{prop:vec}. Fix an interval $I$
and the numbers $N_j$; let $m \geq 2$, and assume, as the induction hypothesis, that
\begin{equation}\label{eq:induction_hyp}
\mathbb{E} \prod_{\ell=0}^{m-2} (\mathcal{N}(H, I)-\ell) \leq \left(C_1
\sum_{j=1}^k N_j \, |I|\right)^{m-1}~,
\end{equation}
for any deformed block-Gaussian random matrix $H$ of the
form \eqref{eq:block_Gaussian_model} in either the orthogonal case
or the unitary case. Note that the induction base,
(\ref{eq:induction_hyp}) with $m=2$, follows from Theorem~\ref{thm:weg}.

Let $H$ be a random matrix of the form
\eqref{eq:block_Gaussian_model} in either the orthogonal case or the
unitary case. The
formula~\eqref{eq:formula} applied to $\mathcal{N}(H, I)$ shows that
\begin{equation*}
  \prod_{\ell=0}^{m-1} (\mathcal{N}(H, I)-\ell) = |I|\lim_{\eta \to +0} \sum_{j=1}^k \operatorname{Ave}_{\lambda,t,\xi}^{\eta}\,
\frac{1}{2\xi} \mathcal{N}\Big(V(j) + A(j,\lambda,\eta,t), (-\xi,
\xi)\Big)\prod_{\ell=1}^{m-1} (\mathcal{N}(H, I)-\ell)~.
\end{equation*}
Thus, by the Fatou lemma, it suffices to prove that for any
$1\le j\le k, \lambda,t\in\mathbb{R}$ and $\xi,\eta>0$,
\begin{equation}\label{eq:Minami_first_reduction}
  \mathbb{E}\, \mathcal{N}\Big(V(j) + A(j,\lambda,\eta,t), (-\xi, \xi)\Big)\prod_{\ell=1}^{m-1} (\mathcal{N}(H, I)-\ell)_{+} \le 2\xi\cdot C_1 N_j \left(C_1
\sum_{i=1}^k N_i \, |I|\right)^{m-1}.
\end{equation}
The eigenvalues of $V(j)+A(j,\lambda,\eta,t)$ are simple almost surely, since
the distribution of $V(j) + A(j,\lambda,\eta,t)$ is absolutely continuous with respect to
the Lebesgue measure on real symmetric matrices (orthogonal case)
or Hermitian matrices (unitary case).  For each natural $M$, construct a partition $\{ I_{M}^n\}_{n=1}^{2^M}$ of $(-\xi, \xi)$ into $2^M$ intervals of equal length. Then, almost surely,
  \begin{equation*}\begin{split}
    \mathcal{N}\Big(V(j) + A(j,\lambda,\eta,t), (-\xi, \xi)\Big)
     &= \lim_{M\to\infty} \sum_{n=1}^{2^M} \mathbbm{1}\left\{\mathcal{N}\Big(V(j) + A(j,\lambda,\eta,t),\, I_{M}^n\Big)\ge 1\right\} \\
&\hspace{-1cm}= \lim_{M\to\infty} \sum_{n=1}^{2^M} \mathbbm{1}\left\{\|(V(j) + A(j,\lambda,\eta,t) - \mathcal{M}(I_{M}^n))^{-1}\|_{\operatorname{op}}\ge \frac{2}{|I_{M}^n|}\right\}~,
\end{split}
  \end{equation*}
 where we denoted by $\mathcal{M}(J)$  the mid-point of an interval $J\subset\mathbb{R}$.
 This relation, combined with the monotone convergence theorem (as the partitions are refining when $M$ increases), reduces the
desired \eqref{eq:Minami_first_reduction} to the following claim: for any  interval $J$,
  \begin{equation}\label{eq:operator_norm_reduction}
    \mathbb{E}\, \mathbbm{1}\left\{\|B_{j,J}^{-1}\|_{\operatorname{op}}\ge \frac{2}{|J|}\right\}\prod_{\ell=1}^{m-1} (\mathcal{N}(H, I)-\ell)_{+} \le |J|\cdot C_1 N_j \left(C_1
\sum_{i=1}^k N_i \, |I|\right)^{m-1},
  \end{equation}
  where we denoted
  \begin{equation*}
    B_{j,J}:=V(j) + A(j,\lambda,\eta,t) - \mathcal{M}({J})~.
  \end{equation*}
  Now let $v$ be a random vector, independent of $H$, which is uniformly distributed on the sphere $\mathbb{S}^{N_j-1}_{\mathbb{R}}$ in the orthogonal case
  or uniformly distributed on the complex sphere $\mathbb{S}^{N_j-1}_{\mathbb{C}}$ in the unitary
  case. By first conditioning on $H$, inequality \eqref{eq:expectation_via_uniform_vector} may be applied to show that
  \begin{equation}\label{eq:random_vector_introduction}\begin{split}
  &\mathbb{E}\, \mathbbm{1}\left\{\|B_{j,J}^{-1}\|_{\operatorname{op}}\ge \frac{2}{|J|}\right\}\prod_{\ell=1}^{m-1} (\mathcal{N}(H, I)-\ell)_{+}\\
  &\le 2\mathbb{E}\left[\mathbbm{1}\left\{\|B_{j,J}^{-1}\|_{\operatorname{op}}\ge \frac{2}{|J|}\right\}\prod_{\ell=1}^{m-1} (\mathcal{N}(H, I)-\ell)_{+}\cdot \mathbbm{1}\left\{\|B_{j,J}^{-1}v\|\ge \frac{\|B_{j,J}^{-1}\|_{\operatorname{op}}}{10\sqrt{N_j}}\right\}\right]\\
  &\le 2\mathbb{E}\left[\mathbbm{1}\left\{\|B_{j,J}^{-1}v\|\ge \frac{1}{5\sqrt{N_j}|J|}\right\}\prod_{\ell=1}^{m-1} (\mathcal{N}(H, I)-\ell)_{+}\right].
  \end{split}\end{equation}
Denote by $P_j: \mathbb{R}^{\sum_i N_i} \to \mathbb{R}^{N_j}$ the coordinate projection
to the space corresponding to $V(j)$, and, for $\tau \in \mathbb{R}$, define the rank-one perturbation
\[ H_{v,\tau} = H + \tau P_j^* v v^* P_j~. \]
The eigenvalues of $H$ and $H_{v,\tau}$ interlace, therefore
\begin{equation}\label{eq:rk1}\prod_{\ell=1}^{m-1} (\mathcal{N}(H, I)-\ell)_{+} \le \mathfrak{P} := \lim_{\tau \to +\infty} \prod_{\ell=0}^{m-2} (\mathcal{N}(H_{v,\tau}, I)-\ell)\end{equation}
(the inequality actually holds for any fixed $\tau$).
In view of  \eqref{eq:random_vector_introduction}, our goal \eqref{eq:operator_norm_reduction} is reduced to the inequality
 \begin{equation}\label{eq:sufficient_bound_with_vec}    2\mathbb{E}\left[\mathbbm{1}\left\{\|B_{j,J}^{-1}v\|\ge \frac{1}{5\sqrt{N_j}|J|}\right\}\mathfrak{P}\right] \leq  |J|\cdot C_1 N_j \left(C_1
\sum_{i=1}^k N_i \, |I|\right)^{m-1}~,  \end{equation}
  which we now prove.

The following simple fact is central to the argument. For an
Hermitian matrix $K$ of dimension $r$ and unit vector $u \in
\mathbb{C}^r$, define a matrix $K_u$ of dimension $r-1$ by $K_u =
\hat{P}_{u^\perp} K_u \hat{P}_{u^\perp}^*$, where
$\hat{P}_{u^\perp}: \mathbb{C}^r \to \mathbb{C}^r / \mathbb{C}u$ is
the canonical projection (for example, if $u$ is the first vector of
the standard basis, $K_u$ is the submatrix obtained by removing the
first row and column of $K_u$). Then
\[  \lim_{\tau \to \infty} \mathcal{N}(K + \tau u u^*, I) = \mathcal{N}(K_u, I) \]
for any interval $I$ whose endpoints are not eigenvalues of $K_u$.
We apply this identity with $K = H$ and $u = P^*_j v$, and deduce
that the random variable $\lim\limits_{\tau \to +\infty}
\mathcal{N}(H_{v,\tau}, I)$ is measurable with respect to
$\hat{P}_{(P_j^* v)^\perp} H \hat{P}_{(P_j^* v)^\perp}^*$. Thus, the
``moreover'' part of Proposition~\ref{prop:vec} can be applied,
yielding
  \begin{equation}\label{eq:cond_vec_application}
  \begin{split}
&2\mathbb{E}\left[\mathbbm{1}\left\{\|B_{j,J}^{-1}v\|\ge \frac{1}{5\sqrt{N_j}|J|}\right\}\mathfrak{P} \right]\\
&=2\mathbb{E}
\left(\mathfrak{P}\cdot\mathbb{P}\left[\|B_{j,J}^{-1}v\|\ge
\frac{1}{5\sqrt{N_j}|J|} \,\,\Bigm|\,\, \hat{P}_{(P_j^* v)^\perp} H
\hat{P}_{(P_j^* v)^\perp}^* \right]\right) 
     \leq  C_1 N_j|J| \mathbb{E}\mathfrak{P}
     \end{split}
  \end{equation}
 with $C_1 = 10C$.
Now note that each of the matrices $H_{v,\tau}$, conditioned on $v$,
has the form \eqref{eq:block_Gaussian_model} in the orthogonal or
unitary case (corresponding to the case of $H$). Thus we may apply
the Fatou lemma and the induction hypothesis
\eqref{eq:induction_hyp} to conclude that
 \begin{equation}\label{eq:fromind}
\mathbb{E} \mathfrak{P} \leq \lim_{\tau \to \infty} \mathbb{E} \prod_{\ell=0}^{m-2} (\mathcal{N}(H_{v,\tau}, I)-\ell) \leq  \left(C_1
\sum_{i=1}^k N_i \, |I|\right)^{m-1}~.
\end{equation}
The combination of (\ref{eq:cond_vec_application}) with (\ref{eq:fromind}) concludes
the proof of (\ref{eq:sufficient_bound_with_vec}) and of the theorem.
\end{proof}

\begin{proof}[Proof of Theorem~\ref{thm:resolvent}] We consider in parallel the cases
of orthogonal and unitary symmetry. First, the matrix $H_L$ is of
the form
$$H_L \ = \ \left( \begin{array}{cccccccc}
V_1 & T_1 & 0 &  \cdots & \cdots & \cdots & 0\\
T_1^* & V_2 & T_2 & 0 & \cdots &\cdots & 0 \\
0 & T_2^* & V_3 & \ddots & \ddots && \vdots \\
\vdots & 0 & \ddots & \ddots &\ddots &\ddots & \vdots \\
\vdots & \vdots & \ddots &\ddots & \ddots &\ddots & 0 \\
\vdots & \vdots & &\ddots &\ddots & \ddots& T_{k-1}\\
0 & 0 & \cdots & \cdots & 0 & T_{k-1}^* &    V_k
\end{array}\right)
$$
with $V_j$, $j=1,\ldots,k$, being $W\times W$ matrices drawn from
the GOE (GUE) and $T_j$, $j=1,\ldots,k-1$, being lower triangular
real (complex) Gaussian matrices.  The individual matrices are
identically distributed within each family and stochastically
independent (within each family and between the families). The
matrix dimension is $2L+1=kW$. Therefore we are in the setting of
\cite[Section~3]{Sch1}.

For the rest of the proof, we fix an arbitrary $t \in (s, 1)$ (its value only affects the constants in
the estimates). According to \cite[Theorem 6]{Sch1}, there exist
$C, \mu>0$ such that for any $i,j\in\{-L,-L+1,\ldots,L\}$,
\begin{equation}\label{eq:est_j} \mathbb{E}\left ( \left | (H_L - \lambda)^{-1}(i,j)\right |^s \right ) \ \le \ C M(W,t)^\frac{s}{t}  e^{-\mu W^{-2\nu-1} |i-j|},\end{equation}
where
$$M(W,t) = \max_{-L\le i,j\le L} \mathbb{E} \left ( \left | (H_L - \lambda)^{-1}(i,j)\right |^t \right ),$$
and $$\nu \ \ge\ \max ( 2, \zeta + \max(a,1+\sigma+2b) )$$ with
$\zeta$, $a$, $\sigma$ and $b$ certain exponents related to the
distribution of the blocks of $H_L$. We refer to \cite{Sch1} for
the definition and discussion of the exponents $\zeta$, $a$ and $b$.
As explained in \cite[Section 5]{Sch1}, for the Gaussian matrices
considered here we can take $\zeta=2$, $a=0$ and $b=0$.

The key improvement afforded in the present work comes from the
exponent $\sigma$, which is related to the Wegner estimate, namely
$\sigma$ is such that for any $R > 1$, real symmetric (Hermitian) $W
\times W$ matrices $A, B$ and a real (complex) arbitrary $W \times
W$ matrix $D$,
$$\mathbb{P} \left \{\left \| (V-A)^{-1} \right \|>R W^{1+\sigma} \right \} \ \le \ \kappa \frac{1}{R}$$
and
$$\mathbb{P} \left \{\left \| \left( \begin{array}{cc}
V-A  & D \\ D^* & V'-B\end{array} \right)^{-1} \right \|>R
W^{1+\sigma} \right \} \ \le \ 2 \kappa \frac{1}{R}~,$$ where $V,
V'$ are independent $W\times W$ random matrices with the GOE (GUE)
distribution. Theorem~\ref{thm:weg}, applied with one or two
diagonal blocks ($k = 1,2$), ensures that these estimates hold with
$\sigma=0$ (in the single block case it suffices to use the result
of \cite{APSSS} stated here as Proposition~\ref{prop:vec}).
According to a Wegner-type estimate of \cite[Theorem II.1]{AM},
\[ M(W, t) \leq C_t W^{\frac{t}{2}}~. \]
Plugging this estimate into (\ref{eq:est_j}) with $\sigma=0$, we
obtain the claim.
\end{proof}

\section{Proof of Proposition \ref{l:formula} }

We start with a preparatory lemma which is a consequence of the
Poisson integral formula.
\begin{lemma}\label{l:poisson}
Let $X, Y$ be Hermitian matrices such that $Y$ is negative
semi-definite, and let $\eta > 0$. Then
\[ \Im \operatorname{tr} (X + iY - i\eta)^{-1}
= \int_{-\infty}^\infty \frac{dt}{\pi(1+t^2)}\Im \operatorname{tr}
(X + tY - i\eta)^{-1}~. \]
\end{lemma}
\begin{proof}
Consider the function
\[ \phi(\xi) = \mathrm{tr} (X + \xi Y - i\eta)^{-1}~, \quad \xi \in \mathbb{C}~, \quad \Im \xi \geq 0~. \]
Observe that, for $\xi$ as above and any non-zero vector $\psi \in \mathbb{C}^N$,
\[ \Im \langle (X + \xi Y - i\eta) \psi, \psi \rangle \le -\eta\|\psi\|^2< 0\]
(where  $N$ is the common dimension of $X$ and $Y$, and
$\langle \cdot,\cdot \rangle$ is the scalar product on $\mathbb{C}^N$). By an elementary
linear-algebreaic argument, $\phi$ is holomorphic in its domain of definition and, in particular,
$\Im \phi$ is harmonic. Also,
\[ \Im \phi(\xi) = \frac{1}{2i} \tr \left\{ (X + \xi Y - i\eta)^{-1} - (X + \bar{\xi} Y + i\eta)^{-1} \right\}   \]
is positive, and
\[ \limsup_{y \to + \infty} \Im \phi(iy) \leq \sup_{y > 0} \Im \phi(iy) \leq N\eta^{-1} < \infty~, \]
since $\| (X + iyY - i\eta)^{-1}\|_{\operatorname{op}} \leq
\frac{1}{\eta}$. Therefore (see e.g.\ \cite[Chapter 14]{Lev}) $\Im
\phi$ admits the Poisson representation
\[ \Im \phi(i) = \int_{-\infty}^{+\infty} \frac{\Im \phi(t)dt}{\pi(1+t^2)} =
\int_{-\infty}^{+\infty} \frac{dt}{\pi(1+t^2)}\Im \mathrm{tr}(X + t
Y - i\eta)^{-1}~.\qedhere\]
\end{proof}

We proceed with the proof of Proposition~\ref{l:formula}. Let $H$
have the form $H = H_0 + \oplus_{j=1}^k V(j)$ (as in
\eqref{eq:block_Gaussian_model}) in which $H_0$ and all $V(j)$ are
(deterministic) Hermitian matrices and suppose that $V(j)$ is of
size $N_j\times N_j$. Denote by $P_j: \mathbb{R}^{\sum_i N_i} \to
\mathbb{R}^{N_j}$ the coordinate projection to the space
corresponding to $V(j)$; also denote by $Q_j: \mathbb{R}^{\sum_i
N_i} \to \mathbb{R}^{\sum_{i \neq j} N_i}$ the coordinate projection
to the orthogonal subspace to the range of $P_j$. Let
\begin{equation}\label{eq:A_j_lambda_eta_t_def_1}
A(j) = P_j H_0 P_j^*~, \quad B(j) = Q_j H P_j^*~, \quad C(j) = Q_j H
Q_j^*~.
\end{equation}
(note that $A(j)$ is defined with $H_0$ rather than $H$) and
define, for $\lambda, t\in\mathbb{R}$ and $\eta>0$,
\begin{equation}\label{eq:A_j_lambda_eta_t_def_2}
A(j,\lambda,\eta,t) =  - \lambda +A(j) - B(j)^* (C(j) - \lambda +
t\eta) ((C(j) - \lambda)^2 +\eta^2)^{-1} B(j)~.
\end{equation}

The Perron--Stieltjes inversion formula \cite[Addenda to Chapter
III]{Akh}, using our assumption that the endpoints of $I$ are not
eigenvalues of $H$, implies that
\begin{equation}\label{eq:Perron_Stieltjes}
\mathcal{N}(H, I) = \lim_{\eta \to +0} \int_I \frac{d\lambda}{\pi}
\Im \operatorname{tr} (H-\lambda-i\eta)^{-1}~.
\end{equation}
Now, for any $\lambda\in\mathbb{R}$ and $\eta>0$ the integrand may
be rewritten using the Schur--Banachiewicz inversion formula,
\begin{equation*}
\begin{split}
\tr (H - \lambda-i\eta)^{-1} &= \sum_{j=1}^k  \tr P_j (H-\lambda-i\eta)^{-1} P_j^*\\
&= \sum_{j=1}^k \tr\Big(V(j) - \lambda-i\eta + A(j) - B(j)^* (C(j) -
\lambda-i\eta)^{-1} B(j)\Big)^{-1}~.\end{split}
\end{equation*}
This expression, in turn, may be rewritten as follows. Denoting, for
each $1\le j\le k$,
\begin{equation*}
  Z(j) = -\lambda + A(j) - B(j)^* (C(j) - \lambda -i\eta)^{-1}
  B(j)~,
\end{equation*}
we may define the Hermitian matrices
\begin{equation*}
\begin{split}
  &X(j) = \frac{Z(j)+Z(j)^*}{2} = - \lambda + A(j) - B(j)^* (C(j) - \lambda) ((C(j) - \lambda)^2 +\eta^2)^{-1}
  B(j)~,\\
  &Y(j) = \frac{Z (j)- Z(j)^*}{2i} = - \eta B(j)^* ((C(j) - \lambda)^2 +\eta^2)^{-1}B(j)~
\end{split}
\end{equation*}
and conclude that
\begin{equation*}
\tr (H - \lambda-i\eta)^{-1} = \sum_{j=1}^k \tr (V(j) + X(j) + iY(j)
- i\eta )^{-1}~.
\end{equation*}
The matrix $Y(j)$ is explicitly negative semi-definite,
therefore Lemma~\ref{l:poisson}  implies that
\begin{equation*}
\begin{split}
  \Im\tr (H - \lambda-i\eta)^{-1} &= \sum_{j=1}^k \int_{-\infty}^\infty \frac{dt}{\pi(1+t^2)}\Im \operatorname{tr} (V(j) + X(j) + tY(j) - i\eta)^{-1}\\
  &= \sum_{j=1}^k \int_{-\infty}^\infty \frac{dt}{\pi(1+t^2)}\Im \operatorname{tr} (V(j) + A(j,\lambda,\eta,t) - i\eta)^{-1}~.
\end{split}
\end{equation*}
Plugging this equality into \eqref{eq:Perron_Stieltjes} shows that
\begin{equation}\label{eq:number_of_eigenvalues_near_final_form}
  \mathcal{N}(H, I) = \lim_{\eta \to +0} \sum_{j=1}^k\int_I \frac{d\lambda}{\pi} \int_{-\infty}^\infty \frac{dt}{\pi(1+t^2)}\Im \operatorname{tr} (V(j) + A(j,\lambda,\eta,t) - i\eta)^{-1}~.
\end{equation}
To conclude the proof of the proposition, it remains to note that
\begin{equation*}
  \Im \operatorname{tr} (V(j) + A(j,\lambda,\eta,t) - i\eta)^{-1} = \int_0^\infty
d\mathcal{N}(V(j) + A(j,\lambda, \eta, t), (-\xi, \xi) )
\frac{\eta}{\xi^2 + \eta^2}~,
\end{equation*}
where the interior integral is a Stieltjes integral with respect to
the variable $\xi$, and, as before,  $\mathcal{N}(V(j) +
A(j,\lambda, \eta, t), (-\xi, \xi))$ is equal to the number of
eigenvalues of  $V(j) + A(j,\lambda, \eta, t)$ in the interval
$(-\xi,\xi)$. Integrating by parts, we obtain:
\[\int_0^\infty d\mathcal{N}(V(j)+ A(j,\lambda, \eta, t), (-\xi, \xi) ) \frac{\eta}{\xi^2 + \eta^2} =
\int_0^\infty \mathcal{N}\Big(V(j) +A(j,\lambda,\eta,t), (-\xi,
\xi)\Big) \frac{2\eta \xi \, d\xi}{(\xi^2 + \eta^2)^2}~.\] The last
two displayed equations together with
\eqref{eq:number_of_eigenvalues_near_final_form} establish the
proposition.

\section{Concluding remarks}\label{s:rmks}
\paragraph{Second-order perturbation theory}
One heuristic explanation for the scaling (\ref{eq:conj}) is provided
by second-order perturbation
theory. We sketch the argument for
$H^\text{bA}$; similar considerations apply  to the Wegner $N$-orbital operator $H^\text{Weg}$.

At  $g=0$, the coupling between blocks is completely suppressed; and  the operator  has pure point spectrum, with eigenvalues given by the union of the spectrum of the individual matrices $V(x)$, $x\in \mathbb{Z}^d$. Let $\lambda_j(x)$, $j=1,\ldots,N$, denote the eigenvalues associated to the matrix $V(x)$, with corresponding eigenvectors $\mathbf{v}_j(x)$, $j=1,\ldots,N$, in $\mathbb{C}^N$.
Then, for every $x$, the distribution of the eigenvalues
is approximately given by Wigner's semicircle density
$(2\pi)^{-1} \sqrt{(4-\lambda^2)_+}$,
and the gaps between the eigenvalues (for fixed $x$) are typically of order $N^{-1}$.

For positive $g = \frac{a}{\sqrt{N}}$, second-order perturbation theory
predicts that the eigenvalues $\lambda_j(x)$ shift by
a quantity close to
\[ \frac{a^2}{N} \sum_{y\sim x} \sum_{k=1}^N \frac{ \left | \left < \mathbf{v}_k(y),  \mathbf{v}_j(x) \right > \right |^2 } { \lambda_j(x) - \lambda_k(y)} \ \approx \
\frac{a^2d }{\pi N} \operatorname{P\!.\!V\!.}\!  \int_{-2}^{2} \frac{\sqrt{4-\lambda^2} }{\lambda_j(x) -\lambda} d \lambda  = \frac{a^2d \,\, \lambda_j(x)}{N}\ ,  \]
i.e.\ comparable to the mean gap.

Though the series in $a$ provided by Rayleigh--Schr\"odin\-ger (infinite order) perturbation theory has zero radius of convergence, the considerations of the previous paragraph provide an indication that the scaling $g = a/\sqrt{N}$ is natural.

\paragraph{Supersymmetric models}
Another perspective on the models (\ref{eq:def}), and random operators in general, is given by dual supersymmetric
models, which were introduced by Efetov \cite{Ef}, following earlier
work by Wegner and Sch\"afer \cite{W1,SW}; see further
the monograph of Wegner \cite{Wegbook} and the
mathematical review of  Spencer  \cite{Sp}.  In the supersymmetric
approach,  $\mathbb{E} |(H-z)^{-1}(x, y)|^2$ is expressed
as a two-point correlation in a dual supersymmetric model. Fixing $a>0$ and setting
$g = \frac{a}{\sqrt{N}}$, the  supersymmetric models dual
to (\ref{eq:def}) should  converge, as $N \to \infty$,
to a supersymmetric $\sigma$-model with $U(1, 1\mid 2)$ symmetry (in the unitary case) and $OSp(2,2 \mid 4)$ symmetry (in the orthogonal case), at  temperature
determined by $a$ and the value of the density of states. These $\sigma$-models are conjectured to exhibit a phase transition
in dimension $d \geq 3$. This provides additional support
for the scaling (\ref{eq:conj}).

As to rigorous results, a mathematical proof of the existence
of phase transition for the supersymmetric $\sigma$-models
remains a major challenge. Progress was made by
Disertori, Spencer, and Zirnbauer \cite{DSZ,DS},  who rigorously established  the existence of phase transition
for a supersymmetric $\sigma$-model with the simpler $OSp(2 \mid 2)$ symmetry.
Presumably, the analysis of the Efetov $\sigma$-models
 presents additional challenges.

The convergence of the dual supersymmetric models
to the corresponding $\sigma$-models is, to the best of our knowledge,
not yet mathematically established. Moreover, a strong form
of convergence is required to infer the existence of a phase transition
before the limit; convergence of the action
does not by itself suffice.

\paragraph{Lower bounds on the density of states}

The upper bound in Theorem~\ref{thm:weg} is often sharp up to a multiplicative
constant. One can obtain complementary bounds to Theorem~\ref{thm:weg} in
terms of the second moment
\begin{equation*}
s_2^2:=\frac{1}{\sum_{j=1}^k N_j} \mathbb{E} \tr H^2 =
\frac{1}{\sum_{j=1}^k N_j}\mathbb{E}\int \lambda^2
\mathcal{N}(H,d\lambda)~.
\end{equation*}
Namely, for any $0<t<s_2$ there exists an interval
$I\subset[-2 s_2,2 s_2]$ with $|I| = t$ for which
\begin{equation}\label{eq:lowerbd}
\mathbb{E} \mathcal{N}(H, I)  \geq \frac{1}{10 s_2} \sum_{j=1}^k N_j \,
|I|~,
\end{equation}
since the failure of this bound would imply that
\[ \mathbb{E} \mathcal{N}(H, [-2 s_2, 2 s_2]) \leq \left\lceil \frac{4s_2}{t} \right\rceil 
\frac{t }{10 s_2} \sum_{j=1}^k N_j \leq \frac12 \sum_{j=1}^k N_j
\]
in contradiction to Chebyshev's inequality.
In the applications to the orbital
models \eqref{eq:def} and to Gaussian band matrices (e.g.\
Definition~\ref{defin:band} with $\psi$ as in (\ref{eq:ekband}) and $\phi$
decaying sufficiently fast),
the quantity $s_2$ is itself bounded by a constant, whence
Theorem~\ref{thm:weg} is sharp in these cases.

It is plausible that, for orbital models and band matrices, bounds
of the type (\ref{eq:lowerbd}) also hold for individual intervals sufficiently close to the origin; see Wegner \cite{W3} for the case $N=1$.

\paragraph{Three open questions} We use this opportunity to recapitulate
a few of the open questions pertaining to the block Anderson and
Wegner orbital models (\ref{eq:def}).
\begin{enumerate}
\item Is it true that in  dimension $d \geq 3$  one has the following converse to 
Theorem~\ref{thm:loc}: for $g \geq C(d) N^{-1/2}$, there exist energies $\lambda$
at which  exponential decay of the form  (\ref{eq:expdecay}) does not hold, at least, for 
 large $N$? Absence of exponential decay for an interval of energies 
 could be considered as a signature of delocalisation.
\item Consider the case of fixed $g$ and $N \to \infty$. Is it true that the density of states, the density of the measure (\ref{eq:defdos}), converges,
as $N \to \infty$, in uniform metric? Convergence in the weak-$*$ metric (to an explicit
limiting measure) was proved  by Khorunzhiy and Pastur in \cite{KhP,P2}. To 
upgrade their result to uniform convergence, it would suffice (by a 
compactness argument) to show that the density of states is equicontinuous in $N$ as a function of the spectral parameter $\lambda$.
\item Is the density of states a smooth function
of the spectral parameter? It is expected to be analytic for all values of $g >0$ and 
$N \geq 1$. See further the results of \cite{CFGK,C} discussed after the proof of Corollary~\ref{cor:orbital_Wegner}, and the  work  \cite{CFS} and references therein pertaining to the Anderson model ($N=1$). 
\end{enumerate}

\paragraph{Acknowledgment.} Tom Spencer encouraged us
to look for a proof to the estimate (\ref{eq:band_matrix_Wegner}). We had the pleasant
opportunity to discuss various aspects of supersymmetric
$\sigma$-models with him as well as with Yan Fyodorov and Tanya Shcherbina.
The current project profited greatly from the helpful advice of Michael Aizenman,
and also crucially relies on the joint work \cite{APSSS}. We thank them
very much.


\begin{thebibliography}{99}

\bibitem{AAT} R.\ Abou-Chacra, D.\ J.\ Thouless, P.\ W.\ Anderson,
A selfconsistent theory of localization,
Journal of Physics C: Solid State Physics 6, no.~10 (1973): 1734.

\bibitem{A} M.~Aizenman,
Localization at weak disorder: some elementary bounds,
Rev.\ Math.\ Phys.~6 (1994), no.~5A, 1163--1182.

\bibitem{AM} M.~Aizenman, S.~Molchanov,
Localization at large disorder and at extreme energies: an elementary derivation,
Comm.\ Math.\ Phys.~157 (1993), no.~2, 245--278.

\bibitem{APSSS} M.~Aizenman, R.~Peled, J.~Schenker, M.~Shamis,  S.~Sodin,
Matrix regularizing effects of Gaussian perturbations,
arXiv:1509.01799.

\bibitem{ASFH} M.~Aizenman, J.~Schenker, R.~Friedrich,
D.~Hundertmark,
Finite-volume fractional-moment criteria for Anderson localization,
Comm.\ Math.\ Phys.~224 (2001), no.~1, 219--253.

\bibitem{AW} M.\ Aizenman, S.\ Warzel,
Random operators.
Disorder effects on quantum spectra and dynamics.
Graduate Studies in Mathematics, 168.
American Mathematical Society, Providence, RI, 2015. xiv+326 pp.

\bibitem{Akh} N.\ I.\ Akhiezer,
The classical moment problem and some related questions in analysis.
Translated by N.\ Kemmer. Hafner Publishing Co., New York 1965 x+253 pp.

\bibitem{And} P.~W.~Anderson,
Absence of diffusion in certain random lattices,
Phys.\ Rev.\ 109.5 (1958): 1492.

\bibitem{Bapst}
V.~Bapst, The large connectivity limit of the Anderson model on tree
graphs, J. Math. Phys. {\bf 55} (2014), no.~9, 092101, 20 pp.

\bibitem{BMP} L.~V.~Bogachev, S.~A.~Molchanov, L.~A.~Pastur,
On the density of states of random band matrices,
Mat. Zametki 50 (1991), no.\ 6, 31--42, 157 (Russian); English transl., Math.\ Notes 50 (1991), no. 5--6, 1232--1242 (1992).

\bibitem{Casati:1990p3447}
G.~Casati, L.~Molinari, and F.~Izrailev, Scaling properties of band random
  matrices, Phys.\ Rev.\ Lett.\ 64 (1990), 1851--1854.

\bibitem{Casati:1993p35}
G.~Casati, B.~V.~Chirikov, I.~Guarneri, F.~M.~Izrailev,
Band-random-matrix model for quantum localization in conservative systems,
Phys.\ Rev.\ E.\ 48 (1993), R1613.

\bibitem{CGK}
J.-M. Combes, F. Germinet\ and\ A. Klein, Generalized
eigenvalue-counting estimates for the Anderson model, J. Stat. Phys.
{\bf 135} (2009), no.~2, 201--216.

\bibitem{C} F.~Constantinescu,
The supersymmetric transfer matrix for linear chains with nondiagonal disorder,
J.\ Stat.\ Phys.~50.5--6 (1988): 1167--1177.

\bibitem{CFGK} F.~Constantinescu, G.~Felder, K.~Gaw{\k{e}}dzki, A.~Kupiainen,
Analyticity of density of states in a gauge-invariant model for disordered electronic systems,
J.\ Stat.\ Phys.~48.3--4 (1987): 365--391.

\bibitem{CFS} F.~Constantinescu, J.~Fr\"ohlich, T.~Spencer,
Analyticity of the density of states and replica method for random Schr\"odinger operators on a lattice. J.\ Statist.\ Phys.\ 34 (1984), no. 3--4, 571--596. 

\bibitem{DL}
M.\ Disertori and M.\ Lager,
Density of States for Random Band Matrices in two dimensions,
arXiv:1606.09387


\bibitem{DPS}
M.~Disertori, H.~Pinson, T.~Spencer,
Density of states for random band matrices,
Comm.\ Math.\ Phys. 232.1 (2002): 83--124.

\bibitem{DS} M.~Disertori, T.~Spencer,
Anderson localization for a supersymmetric sigma model,
Comm.\ Math.\ Phys.~300 (2010), no.\ 3, 659--671.

\bibitem{DSZ} M.~Disertori, T.~Spencer, M.~R.~Zirnbauer,
Quasi-diffusion in a 3D supersymmetric hyperbolic sigma model,
Comm.\ Math.\ Phys.\ 300 (2010), no.~2, 435--486.

\bibitem{Ef} K.~B.~Efetov,
Supersymmetry and theory of disordered metals,
Advances in Physics 32.1 (1983): 53--127.

\bibitem{EK2} L.\ Erd\H{o}s, A.\ Knowles,
Quantum diffusion and delocalization for band matrices with general distribution,
Ann.\ Henri Poincar\'e 12 (2011), no.\ 7, 1227--1319.

\bibitem{EYY}
L.~Erd\H{o}s, H.-T.~Yau, J.~Yin,
Bulk universality for generalized Wigner matrices,
Probability Theory and Related Fields 154.1--2 (2012): 341--407.

\bibitem{FS} J.~Fr\"ohlich, T.~Spencer,
Absence of diffusion in the Anderson tight binding model for large disorder or low energy,
Comm.\ Math.\ Phys.~88 (1983), no.~2, 151--184.

\bibitem{FM} Y.~V.~Fyodorov and A.~D.~Mirlin, Scaling properties of localization in random
  band matrices: A $\sigma$-model approach, Phys.\ Rev.\ Lett.\ 67 (1991), 2405--2409.

\bibitem{KMP}
A.\ M.\ Khorunzhi\v{i} [Khorunzhiy], S.\ A.\ Molchanov,
L.\ A.\ Pastur,
Distribution of the
eigenvalues of random band matrices in the limit of their infinite order,
Teoret. Mat. Fiz.
90 (1992), no. 2, 163--178 (Russian, with English and Russian summaries); English transl.,
Theoret. and Math. Phys. 90 (1992), no. 2, 108--118.


\bibitem{KhP} A.~M.~Khorunzhy [Khorunzhiy], L.~A.~Pastur,
Limits of infinite interaction radius, dimensionality and the number of components for random operators with off-diagonal randomness,
Comm.\ Math.\ Phys.~153.3 (1993): 605--646.

\bibitem{Lev} B.\ Ya.\ Levin,
Lectures on entire functions.
In collaboration with and with a preface by Yu. Lyubarskii, M. Sodin and V. Tkachenko. Translated from the Russian manuscript by Tkachenko. Translations of Mathematical Monographs, 150. American Mathematical Society, Providence, RI, 1996. xvi+248 pp.

\bibitem{Minami}
N. Minami, Local fluctuation of the spectrum of a multidimensional
Anderson tight binding model, Comm. Math. Phys. 177 (1996),
no.~3, 709--725.

\bibitem{OW} R.~Oppermann, F.~Wegner,
Disordered system with $n$ orbitals per site: $1/n$ expansion,
Zeitschrift f\"ur Physik~B: Condensed Matter 34.4 (1979): 327--348.

\bibitem{P2} L.~Pastur,
On connections between the theory of random operators and the theory of random matrices,
St.\ Petersburg Mathematical Journal 23.1 (2012): 117--137.

\bibitem{FP} L.\ Pastur, A.\ Figotin,
Spectra of random and almost-periodic operators. Grundlehren der Mathematischen Wissenschaften [Fundamental Principles of Mathematical Sciences], 297. Springer-Verlag, Berlin, 1992. viii+587 pp.


\bibitem{PShch} L.~Pastur, M.~Shcherbina,
Eigenvalue distribution of large random matrices,
Vol.\ 171. AMS Bookstore, 2011.

\bibitem{Pch} V.~Pchelin,
Poisson statistics for random deformed band matrices with power law band width,
arXiv:1505.06527.

\bibitem{SW} L.~Sch\"afer, F.~Wegner,
Disordered system with $n$ orbitals per site: Lagrange formulation, hyperbolic symmetry, and Goldstone modes,
Zeitschrift f\"ur Physik B: Condensed Matter 38.2 (1980): 113--126.

\bibitem{Sch1} J.~H.~Schenker, Eigenvector localization for random band matrices with power law bandwidth, Comm. Math. Phys. 290 (2009), 1065-1097.

\bibitem{Sch2} J.~H.~Schenker,
How large is large? Estimating the critical disorder for the Anderson model,
Lett. Math. Phys. 105, 1--9.

\bibitem{ShchX2}  M.\ Shcherbina, T.\ Shcherbina,
Transfer matrix approach to 1d random band matrices: density of
states,  J.\ Stat.\ Phys.\ (2016), arXiv:1603.08476.

\bibitem{band2} S.~Sodin,
An estimate for the average spectral measure of random band
matrices, J.\ Stat.\ Phys., Vol.\ 144, Issue 1 (2011), pp. 46--59.

\bibitem{Sp_band} T.~Spencer,
Random banded and sparse matrices,
Oxford Handbook of Random Matrix Theory. Oxford Handbooks in Mathematics. Oxford University Press, Oxford (2011).

\bibitem{Sp} T.~Spencer,
Duality, statistical mechanics, and random matrices,
Current Developments in Mathematics, 2012.

\bibitem{W1} F.~Wegner,
Disordered system with $n$ orbitals per site: $n=\infty$ limit,
Physical Review B 19.2 (1979): 783.

\bibitem{W3} F.~Wegner,
Bounds on the density of states in disordered systems,
Z.~Phys.~B~44 (1981), no.~1--2, 9--15.

\bibitem{Wegbook} F.~Wegner,
Supermathematics and its applications in statistical physics. 
Grassmann variables and the method of supersymmetry. Lecture Notes in Physics, 920. Springer, Heidelberg, 2016. xvii+374 pp. 

\end{thebibliography}
\end{document}